\documentclass[10pt,journal,epsfig]{IEEEtran}

\usepackage[dvips]{graphicx}
\usepackage{graphicx}
\usepackage{amssymb}
\usepackage{cite}
\usepackage{subfigure}
\usepackage{amsmath}
\usepackage{amsthm}
\usepackage{algorithm}
\usepackage{algorithmic}
\usepackage{multirow}

\begin{document}

\title{Compressed Channel Estimation
for IRS-Assisted Millimeter Wave OFDM Systems: A Low-Rank Tensor
Decomposition-Based Approach}

\author{Xi Zheng, Peilan Wang, and Jun Fang, ~\IEEEmembership{Senior Member}
and Hongbin Li,~\IEEEmembership{Fellow,~IEEE}
\thanks{Xi Zheng, Peilan Wang, and Jun Fang are with the National Key Laboratory
of Science and Technology on Communications, University of
Electronic Science and Technology of China, Chengdu 611731, China,
Email: JunFang@uestc.edu.cn}
\thanks{Hongbin Li is
with the Department of Electrical and Computer Engineering,
Stevens Institute of Technology, Hoboken, NJ 07030, USA, E-mail:
Hongbin.Li@stevens.edu}
\thanks{\copyright~2022 IEEE. Personal use of this material
is permitted. Permission from IEEE must be obtained for all other uses, in
any current or future media, including reprinting/republishing this material for
advertising or promotional purposes, creating new collective works, for resale
or redistribution to servers or lists, or reuse of any copyrighted component of
this work in other works.}}

\maketitle

\begin{abstract}
We consider the problem of downlink channel estimation for
intelligent reflecting surface (IRS)-assisted millimeter Wave
(mmWave) orthogonal frequency division multiplexing (OFDM)
systems. By exploring the inherent sparse scattering
characteristics of mmWave channels, we show that the received
signals can be expressed as a low-rank third-order tensor that
admits a tensor rank decomposition, also known as canonical
polyadic decomposition (CPD). A structured CPD-based method is
then developed to estimate the channel parameters. Our analysis
reveals that the training overhead required by our proposed method
is as low as $\mathcal{O}(U^2)$, where $U$ denotes the sparsity of
the cascade channel. Simulation results are provided to illustrate
the efficiency of the proposed method.
\end{abstract}


\section{Introduction}
IRS has emerged as a promising solution to address the blockage
issue and extend the coverage for mmWave communications.
Nevertheless, due to the passive nature of reflecting elements and
the large size of the channel matrix resulting from massive units
at the IRS, channel estimation for IRS-assisted mmWave systems is
very challenging. To reduce the training overhead, some previous
studies exploited the inherent sparsity of mmWave channels and
developed compressed sensing-based methods to estimate the cascade
BS-IRS-user channel \cite{WangFang20,LiuGao20,WeiShen21}. These
works are mainly concerned with the estimation of narrowband
channels. MmWave systems, however, are very likely to operate on
wideband channels with frequency selectivity. As for wideband
channels, the work \cite{WanGao20} considered channel estimation
for IRS-assisted mmWave OFDM systems, where a distributed
orthogonal matching pursuit (OMP) algorithm was proposed by
utilizing the common angular-domain sparsity shared by different
subcarriers. The work \cite{WanGao20}, however, assumes that the
BS-IRS channel is LOS-dominated and known \emph{a priori}.

Recently, some tensor decomposition-based methods, e.g.
\cite{LinJin21,LiHuang21,AraujoAlmeida21}, were proposed for
IRS-assisted systems by exploring the intrinsic multi-dimensional
structure of the received signals. These works, however, did not
utilize the sparse scattering characteristic of the mmWave
channel. In their formulation, the CP rank of the constructed
tensor is equal to the number of reflecting elements at the IRS.
As a consequence, these methods require a training overhead
proportional to the number of reflecting elements, which is
usually large in practice.

In this paper, we develop a new tensor-decomposition channel
estimation method for IRS-assisted mmWave OFDM systems. Different
from \cite{LinJin21,LiHuang21,AraujoAlmeida21}, our work
formulates the received signal as a low-rank third-order tensor by
exploiting the inherent sparse structure of the cascade channel.
The CP rank of the constructed tensor is equal to the sparsity of
the cascade channel. This low-rank structure enables to obtain a
reliable estimate of the cascade channel using only a very small
amount of training overhead. Another challenge of our problem lies
in that, due to the nature of the cascade channel, one of the
factor matrices of the tensor has redundant columns. As a result,
the Kruskal's condition, which is essential to the uniqueness of
the CPD, does not hold and existing CPD-based methods, e.g.
\cite{ZhouFang17}, cannot be applied. To address this difficulty,
in our work, the Vandermonde structure of the factor matrix is
invoked to develop a structured CPD method for channel estimation.

\section{System Model}
Consider an IRS-assisted mmWave OFDM system, where an IRS is
deployed to assist data transmission from the BS to an
omnidirectional-antenna user. For simplicity, we assume that the
direct link between the BS and the user is blocked due to poor
propagation conditions. The total number of OFDM tones
(subcarriers) is $P_0$, among which $P$, say $\{1,2, \cdots ,P\}$,
subcarriers are selected for training. The BS is equipped with a
uniform linear array (ULA) with $N$ antennas and $R$ radio
frequency (RF) chains, where ${R} \ll {N}$. The IRS is a uniform
planar array (UPA) with $M = {M_x} \times {M_y}$ passive
reflecting elements. Each element can independently reflect the
incident signal with a reconfigurable phase shift. Denote
$\mathbf{\Phi} \triangleq \operatorname{diag}(\boldsymbol{v})
\triangleq \operatorname{diag}(e^{j \gamma_{1}}, e^{j \gamma_{2}},
\cdots, e^{j \gamma_{M}})$ as the phase-shift matrix, where
${\boldsymbol{v}\in\mathbb{C}^{M}} $ is the phase-shift vector,
${\gamma _i} \in [ {0,2\pi } ]$ denotes the phase shift
coefficient associated with the $i$th passive reflecting element.

In this paper, we adopt a geometric wideband mmWave channel model
\cite{AlkhateebHeath16} to characterize the channel between the BS
(IRS) and the IRS (user). Specifically, the BS-IRS channel in the
delay domain can be expressed as
\begin{equation}
\boldsymbol{G}(\tau)=\sum_{l=1}^{L} \alpha_{l}
\boldsymbol{a}_{\mathrm{IRS}} \left(\vartheta_{a, l},
\vartheta_{e, l}\right) \boldsymbol{a}_{\mathrm{BS}}^{T}
\left(\phi_{l}\right) \delta\left(\tau-\tau_{l}\right) \\
\end{equation}
where $L$ is the total number of paths between the BS and the IRS,
$\alpha_l$ is the complex gain associated with the $l$th path,
$\phi_{l}$ represents the angle of departure (AoD),
$\{\vartheta_{a,l}, \vartheta_{e,l}\}$ denote the azimuth and
elevation angles of arrival (AoAs), $\tau _l$ denotes the time
delay, $\delta \left( \tau \right)$ denotes the Dirac-delta
function, $\boldsymbol{a}_{{\rm{IRS}}}(\vartheta,\eta)$ and
$\boldsymbol{a}_{\text{BS}}(\phi)$ represent the receive and
transmit array response vectors, respectively. Similarly, the
IRS-user channel in the delay domain is modeled as
\begin{equation}
\boldsymbol{r}(\tau)=\sum_{l=1}^{L_{r}} \varrho_{l}
\boldsymbol{a}_{\mathrm{IRS}}\left(\chi_{a, l}, \chi_{e, l}\right)
\delta\left(\tau-\kappa_{l}\right)
\end{equation}
where $L_r$ is the number of paths between the IRS and the user,
$\varrho_{l}$ denotes the associated complex path gain,
$\{\chi_{a, l}, \chi_{e, l}\}$ denote the azimuth and elevation
angles of departure, and $\kappa_{l}$ is the time delay.

\begin{figure}[t]
\centering
\includegraphics[width=5.5cm]{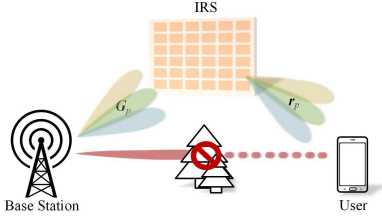}
\caption{IRS-assisted mmWave systems.}
\end{figure}

Accordingly, the frequency-domain BS-IRS and IRS-user channel
matrices associated with the $p$th subcarrier can be respectively
written as
\begin{equation}
\boldsymbol{G}_{p}=\sum_{l=1}^{L} \alpha_{l} e^{-j 2 \pi f_{s}
\tau_{l} \frac{p}{P_{0}}}
\boldsymbol{a}_{\mathrm{IRS}}\left(\vartheta_{a, l}, \vartheta_{e,
l}\right) \boldsymbol{a}_{\mathrm{BS}}^{T}\left(\phi_{l}\right)
\label{eqn1}
\end{equation}
\begin{equation}
\boldsymbol{r}_{p}=\sum_{l=1}^{L_{r}} \varrho_{l} e^{-j 2 \pi
f_{s} \kappa_{l} \frac{p}{P_{0}}}\boldsymbol{a}_{\rm{IR
S}}\left(\chi_{a, l}, \chi_{e, l}\right) \label{eqn2}
\end{equation}
where ${f_s} = 1/{T_s}$ is the sample frequency.

To facilitate the algorithmic development, we consider a
framed-based downlink training protocol. For each subcarrier, the
BS employs $T$ different beamforming vectors at $T$ consecutive
time frames. Each time frame is divided into $Q$ time slots. At
the $q$th time slot, the IRS uses an individual phase-shift
matrix, denoted as $\boldsymbol{\Phi}_q$, to reflect the incident
signal. The beamforming vector associated with the $p$th
subcarrier at the $t$th time frame can be expressed as
$\boldsymbol{x}_p(t)=\boldsymbol{F}_{\text{RF}}(t)\boldsymbol{f}_{\text{BB},p}(t)s_p(t)$,
where $s_p(t)$ denotes the $p$th subcarrier's pilot symbol,
$\boldsymbol{f}_{\text{BB},p}(t)\in\mathbb{C}^{R}$ is the digital
precoding vector for the $p$th subcarrier, and
$\boldsymbol{F}_{\text{RF}}(t)\in\mathbb{C}^{N\times R}$ is a RF
precoder common to all subcarriers. For simplicity, we assume that
$\boldsymbol{f}_{\text{BB},p}(t)=\boldsymbol{f}_{\text{BB}}(t)$
and $s_p(t)=1,\forall p$, in which case we have
\begin{align}
\boldsymbol{x}_p(t)=\boldsymbol{f}(t)\triangleq\boldsymbol{F}_{\text{RF}}(t)\boldsymbol{f}_{\text{BB}}(t),
\forall p
\end{align}
The transmitted signal arrives at the user via propagating through
the BS-IRS-user channel. At the $t$th time frame, the received
signal associated with the $p$th subcarrier at the $q$th time slot
can thus be written as
\begin{align}
y_{p, q}(t) =\boldsymbol{r}_{p}^{T} \boldsymbol{\Phi}_{q}
\boldsymbol{G}_{p} \boldsymbol{f}(t)+n_{p, q}(t)
=\boldsymbol{v}_{q}^T\boldsymbol{H}_{p}\boldsymbol{f}(t)+n_{p,
q}(t) \label{received-signal}
\end{align}
where $\boldsymbol{\Phi}_{q}=\text{diag}(\boldsymbol{v}_{q})$,
$\boldsymbol{H}_{p}\triangleq\text{diag}(\boldsymbol{r}_{p})\boldsymbol{G}_{p}$
denotes the cascade BS-IRS-user channel associated with the $p$th
subcarrier, and ${n_{p,q}}\left( t \right)$ denotes the additive
Gaussian noise.

Substituting (\ref{eqn1})--(\ref{eqn2}) into $\boldsymbol{H}_{p}$,
we arrive at
\begin{align}
\nonumber \boldsymbol{H}_{p}&=\sum_{m=1}^{L_{r}} \sum_{n=1}^{L}
\varrho_{m} \alpha_{n} e^{-j 2 \pi f_{s}
\frac{p}{P_{0}}\left(k_{m}+\tau_{n}\right)}
\\\nonumber
 &\ \ \ \ \ \ \ \ \times\boldsymbol{a}_{\mathrm{IRS}}\left(\chi_{a, m}+
 \vartheta_{a, n}, \chi_{e, m}+\vartheta_{e, n}\right) \boldsymbol{a}_{\mathrm{BS}}^{T}\left(\phi_{n}\right)\\
&\stackrel{(a)}{=} \sum_{u=1}^{L_r L} \beta_{u} e^{-j 2 \pi
f_{s} \frac{p}{P_{0}} \iota_{u}}
\boldsymbol{a}_{\mathrm{IRS}}\left(\omega_{a, u}, \omega_{e,
u}\right) \boldsymbol{a}_{\mathrm{BS}}^{T}\left(\phi_{u}\right)
\label{hpq2}
\end{align}
where the mapping process $\left( a \right)$ is defined as
\begin{gather}
\nonumber (m-1) L+n \mapsto u, u=1, \ldots, L L_{r} \\\nonumber
\varrho_{m} \alpha_{n} \mapsto \beta_{u}, u=1, \ldots, L L_{r}
\\\nonumber \kappa_{m}+\tau_{n} \mapsto \iota_{u}, u=1, \ldots, L
L_{r} \\\nonumber \boldsymbol{a}_{\rm{IR S}}\left(\chi_{a,
m}+\vartheta_{a, n}, \chi_{e, m}+\vartheta_{e, n}\right) \mapsto \
 \boldsymbol{a}_{\mathrm{IRS}}\left(\omega_{a, u}, \omega_{e, u}\right)\\
\boldsymbol{a}_{\mathrm{BS}}\left(\phi_{n}\right) \mapsto
\boldsymbol{a}_{\mathrm{BS}}\left(\phi_{u}\right), n= {\bmod (u,
L)}
\end{gather}
Our objective is to estimate the cascade channel matrices
$\{\boldsymbol{H}_{p}\}$ from the received measurements
$\{y_{p,q}(t)\}$. Note that in the data transmission stage, the
knowledge of $\{\boldsymbol{H}_p\}$ suffices for joint active and
passive beamforming, i.e. optimizing $\boldsymbol{v}$ and
$\{\boldsymbol{f}_p=\boldsymbol{F}_{\text{RF}}\boldsymbol{f}_{\text{BB},p}\}_p$
to maximize the spectral efficiency.

\section{Proposed CPD-Based Method}
\subsection{Low-Rank Tensor Representation}
Substituting (\ref{hpq2}) into (\ref{received-signal}), we obtain
\begin{align}
\nonumber y_{p, q}(t) &= \sum_{u=1}^{L_{r} L} \beta_{u} e^{-j 2
\pi f_{s} \frac{p}{P_{0}} \iota_{u}} \boldsymbol{v}_{q}^{T}
\boldsymbol{a}_{\mathrm{IRS}}
\left(\omega_{a, u}, \omega_{e, u}\right) \\
&\ \ \ \ \ \ \ \ \ \ \ \ \ \ \ \ \ \ \ \ \
\times\boldsymbol{a}_{\mathrm{BS}}^{T}\left(\phi_{u}\right)
\boldsymbol{f}(t)+n_{p, q}(t)
\end{align}
Define $\boldsymbol{y}_{p}(t) \triangleq\left[y_{p,
1}(t)\phantom{0} \cdots\phantom{0} y_{p, Q}(t)\right]^{T} \in
\mathbb{C}^{Q}$. The received signal at the $t$th time frame can
be written as
\begin{align}
\nonumber \boldsymbol{y}_{p}(t)&=\sum_{u=1}^{L_{r} L} \beta
_{u} e^{-j 2 \pi f_{s}
\frac{p}{P_{0}} \iota_{u}} \boldsymbol{V}^{T} \boldsymbol{a}_{\rm{IRS}}\left(\omega_{a, u}, \omega_{e, u}\right) \\
&\ \ \ \ \ \ \ \ \ \ \ \ \ \ \ \
\times\boldsymbol{a}_{\mathrm{BS}}^{T}\left(\phi_{u}\right)
\boldsymbol{f}(t)+\boldsymbol{n}_{p}(t)
\end{align}
where
\begin{gather}
\boldsymbol{V} \triangleq\left[\boldsymbol{v}_{1}\phantom{0}
\cdots\phantom{0}
\boldsymbol{v}_{Q}\right] \in \mathbb{C}^{M \times Q} \\
\boldsymbol{n}_{p}(t) \triangleq\left[n_{p, 1}(t)\phantom{0}
\cdots\phantom{0} n_{p, Q}(t)\right]^{T} \in \mathbb{C}^{Q}
\end{gather}

After receiving signals across all $T$ time frames, the received
signal associated with the $p$th subcarrier can be further
expressed as a matrix
\begin{align}
\boldsymbol{Y}_{p}=\sum_{u=1}^{L_{r} L} \beta _{u} e^{-j 2 \pi
f_{s}\frac{p}{P_{0}} \iota_{u}}
\tilde{\boldsymbol{a}}_{\mathrm{IRS}}\left(\omega_{a, u},
\omega_{e, u}\right)
\tilde{\boldsymbol{a}}_{\mathrm{BS}}^{T}\left(\phi_{u}\right)+\boldsymbol{N}_{p}
\end{align}
where
\begin{gather}
\nonumber \boldsymbol{Y}_{p}
\triangleq\left[\boldsymbol{y}_{p}(1)\phantom{0} \cdots\phantom{0}
\boldsymbol{y}_{p}(T)\right] \in \mathbb{C}^{Q \times T}
\\\nonumber \tilde{\boldsymbol{a}}_{\mathrm{IRS}}\left(\omega_{a,
u}, \omega_{e, u}\right) \triangleq \boldsymbol{V}^{T}
\boldsymbol{a}_{\mathrm{IRS}}\left(\omega_{a, u}, \omega_{e,
u}\right) \in \mathbb{C}^{Q} \\\nonumber
\tilde{\boldsymbol{a}}_{\mathrm{BS}}\left(\phi_{u}\right)
\triangleq \boldsymbol{F}^{T}
\boldsymbol{a}_{\mathrm{BS}}\left(\phi_{u}\right) \in
\mathbb{C}^{T}\\\nonumber \boldsymbol{F}
\triangleq[\boldsymbol{f}(1)\phantom{0} \cdots\phantom{0}
 \boldsymbol{f}(T)] \in \mathbb{C}^{N \times T} \\
\boldsymbol{N}_{p}
\triangleq\left[\boldsymbol{n}_{p}(1)\phantom{0} \cdots\phantom{0}
\boldsymbol{n}_{p}(T)\right] \in \mathbb{C}^{Q \times T}
\end{gather}
As signals from multiple subcarriers are available at the
receiver, the received signal can be expressed as a third-order
tensor $\boldsymbol{\cal Y} \in \mathbb{C}^{Q \times T \times P}$.
It can be readily verified that the tensor $\boldsymbol{\cal Y}$
admits a CPD form as
\begin{align}
\boldsymbol{\cal Y} =\sum_{u=1}^{U}
\tilde{\boldsymbol{a}}_{\mathrm{IRS}}\left(\omega_{a, u},
\omega_{e, u}\right) \circ \left(\beta _{u}
\tilde{\boldsymbol{a}}_{\mathrm{BS}}\left(\phi_{u}\right)\right)
\circ \boldsymbol{g}\left(\iota_{u}\right)+\boldsymbol{\cal N}
\end{align}
where $U\triangleq L L_r$, $\boldsymbol{\cal N} \in \mathbb{C}^{Q
\times T\times P }$ is the tensor representation of the
observation noise, and
\begin{align}
\boldsymbol{g}\left(\iota_{u}\right) \triangleq [e^{-j 2 \pi
\frac{f_{s}}{P_{0}}{\iota _{u}}}\phantom{0} \ldots\phantom{0}
e^{-j 2 \pi \frac{f_{s}}{P_{0}} P{\iota _{u}}}]^{T} \in
\mathbb{C}^{P }
\end{align}

Define
\begin{align}
\boldsymbol{A}\triangleq&\left[\tilde{\boldsymbol{a}}_{\mathrm{IRS}}\left(\omega_{a,
1}, \omega_{e, 1}\right)\phantom{0} \ldots\phantom{0}
\tilde{\boldsymbol{a}}_{\mathrm{IRS}}\left(\omega_{a, U},
\omega_{e, U}\right)\right] \in \mathbb{C}^{Q \times U}
\\
\boldsymbol{B}\triangleq&\left[\beta _{1}
\tilde{\boldsymbol{a}}_{\mathrm{BS}}\left(\phi_{1}\right)\phantom{0}
\ldots\phantom{0}\beta
_{U}\tilde{\boldsymbol{a}}_{\mathrm{BS}}\left(\phi_{U}\right)\right]
\in \mathbb{C}^{T \times U} \\
\boldsymbol{C} \triangleq& \left[
\boldsymbol{g}\left(\iota_{1}\right)\phantom{0} \ldots\phantom{0}
\boldsymbol{g}\left(\iota_{U}\right)\right] \in \mathbb{C}^{P
\times U} \label{factor-matrix-C}
\end{align}
Here $\{\boldsymbol{A},\boldsymbol{B},\boldsymbol{C}\}$ are the
factor matrices of the tensor $\boldsymbol{\cal Y}$. We see that
the channel parameters $\left\{ {{{\omega }_{a,u}},{{\omega
}_{e,u}},{{\phi }_u},{{ \iota }_u},{{ \beta }_u}} \right\}_{u =
1}^{ U}$ can be readily estimated from the factor matrices.
Inspired by this observation, we first estimate the three factor
matrices from the tensor $\boldsymbol{\cal Y}$, and then estimate
the associated channel parameters based on the estimated factor
matrices.

\subsection{Uniqueness Condition}
A well-known sufficient condition for the uniqueness of the CP
decomposition is given in \cite{KruskalJoseph77} and summarized
as
\newtheorem{theorem}{Theorem}
\begin{theorem}
Let ${\boldsymbol{\chi}}\in\mathbb{C}{^{I \times J \times K}}$ be
a third-order tensor decomposed of three factor matrices
${\boldsymbol{A}^{(1)}} \in \mathbb{C}^{I \times R}$,
${\boldsymbol{A}^{(2)}} \in\mathbb{C}^ {{J \times R}}$ and
${\boldsymbol{A}^{(3)}} \in\mathbb{C} {^{K \times R}}$, if the
condition
\begin{align}
k_{\boldsymbol{A}^{(1)}} + k_{\boldsymbol{A}^{(2)}} +
k_{\boldsymbol{A}^{(3)}} \ge 2R + 2 \label{Kruskal-condition}
\end{align}
is satisfied, then the CPD of ${\boldsymbol{{\chi }}}$ is unique
up to scaling and permutation ambiguities. Here
$k_{\boldsymbol{A}}$ denotes the k-rank of $\boldsymbol{A}$, which
is defined as the largest value of $k_{\boldsymbol{A}}$ such that
every subset of $k_{\boldsymbol{A}}$ columns of $\boldsymbol{A}$
is linearly independent.
\end{theorem}

Clearly, the above Kruskal's condition (\ref{Kruskal-condition})
does not hold if $k_{\boldsymbol{A}^{(i)}} = 1,\exists i = 1,2,3$.
Unfortunately, in our problem, the factor matrix $\boldsymbol{B}$
has redundant columns when $L_r \ne 1$, in which case we have
$\boldsymbol{\tilde{a}}_{\mathrm{BS}}(\phi_{u})=\boldsymbol{\tilde{a}}_{\mathrm{BS}}(\phi_{n})$,
for any $u\in \{u: \text{mod}(u,L)=n\}$. Redundant columns
indicate that $k_{\boldsymbol{B}} = 1$. Hence the Kruskal's
condition (\ref{Kruskal-condition}) cannot be satisfied.

To address this difficulty, note that the factor matrix
$\boldsymbol{C}$ (cf. (\ref{factor-matrix-C})) is a Vandermonde
matrix. Previous studies show that, even if the Kruskal's
condition does not hold valid, the CPD is still unique when one of
its factor matrices has a Vandermonde structure. The uniqueness
result was summarized as follows.

\begin{theorem} \label{theorem2}
Let ${\boldsymbol{{\chi }}} \in\mathbb{C} {^{I \times J \times
K}}$ be a third-order tensor decomposed of three factor matrices
${\boldsymbol{A}^{(1)}} \in\mathbb{C} {^{I \times
R}}$,${\boldsymbol{A}^{(2)}} \in\mathbb{C} {^{J \times R}}$ and
${\boldsymbol{A}^{(3)}} \in\mathbb{C} {^{K \times R}}$, where
${\boldsymbol{A}^{(3)}}$ is a Vandermonde matrix with distinct
generators. If the condition
\begin{align}
\begin{cases}
&\mathrm{rank}(\underline{\boldsymbol{A}}^{(3)} \odot \boldsymbol{A}^{(2)})=R \\
&\mathrm{rank}(\boldsymbol{A}^{(1)})=R \label{SCP-condition}
\end{cases}
\end{align}
is satisfied, then the CPD is unique, where
$\underline{\boldsymbol{A}}$ denotes a submatrix of
$\boldsymbol{A}$ that is obtained by removing the bottom row of
$\boldsymbol{A}$, and $\odot$ denotes the Khatri-Rao product.
\end{theorem}
\begin{proof}
See \cite{SorensenLathauwer13}.
\end{proof}

From Theorem \ref{theorem2}, we know that if
\begin{align}
\begin{cases}
&\mathrm{rank}\left(\underline{\boldsymbol{C}} \odot \boldsymbol{B}\right)=U \\
&\mathrm{rank}\left(\boldsymbol{A}\right)=U \label{condition}
\end{cases}
\end{align}
is satisfied and $\boldsymbol{C}$ is a Vandermonde matrix with
distinct generators, then the CP decomposition of
$\boldsymbol{\cal{Y}}$ is unique.

We first examine the rank of $(\underline{\boldsymbol{C}} \odot
\boldsymbol{B})$. Note that the factor matrix $\boldsymbol{C}$ is
a Vandermonde matrix with distinct generators, as we generally
have $\iota_{i}\neq\iota_{j},\forall i\neq j$. Thus the matrix
$(\underline{\boldsymbol{C}} \odot \boldsymbol{B})$ has full
column rank even if $\boldsymbol{B}$ has linearly dependent
columns, provided that $(P-1)T\geq U$ \cite{SorensenLathauwer13}.
On the other hand, recall that the factor matrix $\boldsymbol{A}$
has a form as
\begin{align}
\boldsymbol{A}=\boldsymbol{V}^T
[{\boldsymbol{a}}_{\mathrm{IRS}}\left(\omega_{a,
1},\omega_{e,1}\right)\phantom{0}...\phantom{0}{\boldsymbol{a}}_{\mathrm{IRS}}\left(\omega_{a,
U}, \omega_{e,U}\right)]=
\boldsymbol{V}^T{\boldsymbol{A}_{\mathrm{IRS}}}
\end{align}
Note that $\boldsymbol{A}_{\mathrm{IRS}}$ is a matrix consisting
of a set of steering vectors characterized by different angular
parameters. When entries of $\boldsymbol{V}$ are chosen uniformly
from a unit circle, it is shown in \cite{ZhouFang17} that the
k-rank of $\boldsymbol{A}$ is equal to $\min\{Q,U\}$. When $Q\geq
U$, we have $\mathrm{rank}\left(\boldsymbol{A}\right)=U$.

In summary, conditions (\ref{condition}) are generally satisfied
when $\min \left( {\left( {P - 1} \right)T,Q} \right) \ge U$.
Since the total number of measurements required for our method is
$PTQ$, conditions (\ref{condition}) imply that our proposed method
has a sample complexity of $\mathcal{O}(U^2)$, which only depends
on the sparsity of the cascade channel.

\subsection{CP Decomposition}
We introduce the method \cite{SorensenLathauwer13,LinJin20} to
recover the factor matrices of $\boldsymbol{\cal{Y}}$ by
exploiting the Vandermonde structure inherent in the factor
matrix. Consider the mode-1 unfolding of the received tensor
${\boldsymbol{\cal Y}}$:
\begin{align}
\boldsymbol{Y}_{\left( 1 \right)}^T = \left( {\boldsymbol{C} \odot
\boldsymbol{B}} \right){\boldsymbol{A}^T}+ \boldsymbol{N}_{\left(
1 \right)}^T
\end{align}
and perform the truncated singular value decomposition (SVD)
$\boldsymbol{Y}_{\left( 1 \right)}^T =
\boldsymbol{U}\mathbf{\Sigma} {\boldsymbol{V}^H} \in\mathbb{C}
{^{TP \times Q}}$, where $\boldsymbol{U} \in \mathbb{C}{^{TP
\times U}}$, $\boldsymbol{\Sigma}\in \mathbb{C}{^{ U \times  U}}$
and $\boldsymbol{V}\in \mathbb{C}{^{Q \times U}}$. Here $U$ can be
estimated via a minimum description length (MDL) criterion
\cite{LiuCosta16}.

Ignoring the noise, from (\ref{condition}), we know that there
exists a nonsingular matrix $\boldsymbol{M} \in \mathbb{C}{^{U
\times U}}$ such that
\begin{align}\label{eqn3}
\boldsymbol{U M} = \boldsymbol{C} \odot \boldsymbol{B}
\end{align}
The above equation implies that
\begin{align} \label{eqn4}
\boldsymbol{U}_{1} \boldsymbol{M} &=  \underline{\boldsymbol{C}} \odot \boldsymbol{B} \\
\boldsymbol{U}_{2} \boldsymbol{M} &= \overline{\boldsymbol{C}}
\odot \boldsymbol{B}
\end{align}
where $\overline{\boldsymbol{A}}$ denotes a submatrix of
$\boldsymbol{A}$ obtained by removing the top row of
$\boldsymbol{A}$, and
\begin{align}
\boldsymbol{U}_{1}=\boldsymbol{U}(1:(P-1) T,:) \in \mathbb{C}^{(P-1) T \times U} \\
\boldsymbol{U}_{2}=\boldsymbol{U}(T+1: P T,:) \in
\mathbb{C}^{(P-1) T \times U}
\end{align}
On the other hand, by utilizing the Vandermonde structure of
$\boldsymbol{C}$, we have
\begin{align} \label{eqn5}
\left( \underline{\boldsymbol{C}} \odot \boldsymbol{B}
\right)\boldsymbol{Z} = \overline{\boldsymbol{C}} \odot
\boldsymbol{B}
\end{align}
where $\boldsymbol{Z} \triangleq {\rm{diag}}\left( {{z_1},...,{z_
U}} \right)$, and ${{z_u} \buildrel \Delta \over = {e^{ - j2\pi
{{{f_s}}\over {{P_0}}}{\iota _u}}}}$ is the generator of the
factor matrix $\boldsymbol{C}$. Combining
(\ref{eqn4})--(\ref{eqn5}), we arrive at
\begin{align}
{\boldsymbol{U}_2}\boldsymbol{M} =
{\boldsymbol{U}_1}\boldsymbol{MZ}
\end{align}
Since ${\underline{\boldsymbol{C}} \odot \boldsymbol{B}}$ is full
column rank, both ${\boldsymbol{U}_1}$ and ${\boldsymbol{U}_2}$
are full column rank. Therefore, the generators $\left\{ {{\hat
z_u}} \right\}_{u = 1}^ U$ and $\boldsymbol{\hat M}$ can be
obtained from the eigenvalue decomposition (EVD) of
${\boldsymbol{U}_1^{\dagger} {\boldsymbol{U}_2} = \boldsymbol{\hat
M\hat Z\hat M}^{ - 1}}$. Each column of the factor matrix
$\boldsymbol{C}$ can be estimated as
\begin{align}
{\boldsymbol{\hat c}_u} = \left[ {{{\hat z}_u}\phantom{0}{\hat
z}_u^2\phantom{0} \ldots\phantom{0} {\hat z}_u^P} \right]^T
\end{align}
According to (\ref{eqn3}), the column of the factor matrix
$\boldsymbol{B}$ can be estimated as
\begin{align}
{{\boldsymbol{\hat b}}_u} \buildrel \Delta \over = \bigg(
{{{{\boldsymbol{\hat c}}_u^H} \over {{\boldsymbol{\hat
c}}_u^H{{\boldsymbol{\hat c}}_u}}} \otimes {{\boldsymbol{I}}_T}}
\bigg) {\boldsymbol{U}}{\boldsymbol{\hat M}}{{\left( {:,u}
\right)}}
\end{align}
Finally, given $\boldsymbol{\hat B}$ and $\boldsymbol{\hat C}$,
the factor matrix $\boldsymbol{A}$ can be given as
\begin{equation}
\boldsymbol{\hat A}=\boldsymbol{Y}_{(1)}\left((\boldsymbol{\hat C}
\odot \boldsymbol{\hat B})^{T}\right)^{\dagger}
\end{equation}

\subsection{Channel Estimation}
After obtaining $\hat{\boldsymbol{A}}$, $\hat{\boldsymbol{B}}$ and
$\hat{\boldsymbol{C}}$, we now proceed to estimate the channel
parameters $\{{{{\hat \omega }_{a,u}},{{\hat \omega
}_{a,e}},{{\hat \phi }_u},{{\hat \iota }_u},{{\hat \beta }_u}}
\}_{u = 1}^{U}$. From the above discussion, we know that the
estimated $\{\hat{\boldsymbol{A}}
,\hat{\boldsymbol{B}},\hat{\boldsymbol{C}}\}$ and the true factor
matrices $\{{\boldsymbol{A}} ,{\boldsymbol{B}},{\boldsymbol{C}}\}$
are related as
\begin{align}
\nonumber \hat{\boldsymbol{A}}
&=\boldsymbol{A}\boldsymbol{\Psi}_1\boldsymbol{\Gamma}+\boldsymbol{E}_1
\\\nonumber
\hat{\boldsymbol{B}} &=\boldsymbol{B}\boldsymbol{ \Psi}_2 \boldsymbol{\Gamma}+\boldsymbol{E}_2 \\
\hat{\boldsymbol{C}} &=\boldsymbol{C}
\boldsymbol{\Gamma}+\boldsymbol{E}_3
\end{align}
where $\{\boldsymbol{\Psi}_1, \boldsymbol{\Psi}_2\}$ are
nonsingular diagonal matrices which satisfy $\boldsymbol{\Psi}_1
\boldsymbol{\Psi}_2=\boldsymbol{I}_{U}$, and $\{\boldsymbol{E}_1,
\boldsymbol{E}_2, \boldsymbol{E}_3\}$ are estimation errors.
$\boldsymbol{\Gamma}$ is an unknown permutation matrix. This
permutation matrix $\boldsymbol{\Gamma}$ is common to all factor
matrices, and thus can be ignored.

From an estimated generators $\{\hat{z}_u\}$, the delay parameter
${\hat \iota _u}$ can be estimated as
\begin{align}
{\hat \iota _u} \buildrel \Delta \over =  - {{{P_0}} \over {2\pi
{f_s}}}\arg(\hat{z}_u)
\end{align}
where $\arg(\hat{z}_u)$ denotes the argument of the complex number
$\hat{z}_u$. Recall that each column of the factor matrix
$\boldsymbol{A}$ is characterized by angle parameters
$\{\omega_{a,u},\omega _{e,u}\}$. Therefore these two angle
parameters can be estimated through a correlation-based estimator:
\begin{align}
{\hat \omega _{a,u}},{\hat \omega _{e,u}} = \arg \mathop {\max }
\limits_{{\omega _{a,u}},{\omega _{e,u}}} {{\left|
{{\boldsymbol{\hat a}}_u^H{{{\boldsymbol{\tilde a}}}_{{\rm{IRS}}}}
\left( {{\omega _{a,u}},{\omega _{e,u}}} \right)} \right|} \over
{{{\left\| {{{{\boldsymbol{\hat a}}}_u}} \right\|}_2}{{\left\|
{{{{\boldsymbol{\tilde a}}}_{{\rm{IRS}}}} \left( {{\omega
_{a,u}},{\omega _{e,u}}} \right)} \right\|}_2}}}
\end{align}
where ${{\boldsymbol{\hat a}}_u}$ denotes the $u$th column of
$\hat{\boldsymbol{A}}$. Similarly, the AoD associated with the BS,
can be estimated as
\begin{align}
{\hat \phi _u} = \arg \mathop {\max }\limits_{{\phi _u}} {{|
{{\boldsymbol{\hat b}}_u^H{{{\boldsymbol{\tilde a}}}_{{\rm{BS}}}}
( {{\phi _u}} )} |} \over {{{\| {{{{\boldsymbol{\hat b}}}_u}}
\|}_2}{{\left\| {{{{\boldsymbol{\tilde a}}}_{{\rm{BS}}}}\left(
{{\phi _u}} \right)} \right\|}_2}}}
\end{align}
where ${{\boldsymbol{\hat b}}_u}$ denotes the $u$th column of
$\hat{\boldsymbol{B}}$.

Next, we try to recover the composite path loss gains $\{\hat
{\beta}_u\}$. After obtaining $\{{\hat \omega _{a,u}},{\hat
\omega _{e,u}}\}$, the factor matrix $\boldsymbol{A}$ can be
accordingly estimated as
\begin{gather}
\boldsymbol{\tilde{A}}=\left[\tilde{\boldsymbol{a}}_{\mathrm{IRS}}\left(\hat
\omega_{a, 1}, \hat \omega_{e, 1}\right)\phantom{0}
\ldots\phantom{0} \tilde{\boldsymbol{a}}_{\mathrm{IRS}}\left(\hat
\omega_{a, U}, \hat \omega_{e, U}\right)\right]
\end{gather}
Ignoring the estimation errors, $\boldsymbol{\hat{A}}$ and
$\boldsymbol{\tilde{A}}$ are related as
$\boldsymbol{\hat{A}}=\boldsymbol{\tilde{A}}\boldsymbol{\Psi}_1$.
Hence the nonsingular diagonal matrix $\boldsymbol{\Psi}_1$ can be
estimated as ${\boldsymbol \Psi}_1 = {\tilde {\boldsymbol
A}}^\dagger{\hat{ \boldsymbol A}}$. Since ${{\boldsymbol
\Psi}_1}{{\boldsymbol \Psi}_2}={\mathbf I}_U$, ${\boldsymbol
\Psi}_2$ can be obtained as ${\boldsymbol \Psi}_2 =
\boldsymbol{\Psi}_{1}^{-1}$.

On the other hand, after obtaining ${\hat{\phi}}_{u}$, we can
construct a new matrix
\begin{align} \label{eqn7}
\nonumber \boldsymbol{\tilde{B}}&=[
\tilde{\boldsymbol{a}}_{\mathrm{BS}}( {\hat{\phi}}_{1})\phantom{0}
\ldots\phantom{0}\tilde{\boldsymbol{a}}_{\mathrm{BS}}(
{\hat{\phi}}_{U})]
\end{align}
Ideally we should have
$\boldsymbol{B}=\boldsymbol{\tilde{B}}\boldsymbol{G}$, where
$\boldsymbol{G}\triangleq\text{diag}(\beta _{1},\ldots,\beta
_{U})$. Moreover, ignoring estimation errors, we should have
$\boldsymbol{\hat{B}}=\boldsymbol{B}\boldsymbol{\Psi}_2$.
Therefore $\boldsymbol{G}$ can be estimated as
\begin{gather}
{\hat {\boldsymbol G}} = \boldsymbol{\tilde{B}}^{\dagger} {\hat
{\boldsymbol B}}\boldsymbol{\Psi}_2^{-1}
\end{gather}
Finally, the cascade channels $\{\boldsymbol{H}_p\}$ can be
estimated after those parameters $\{ {{{\hat \omega
}_{a,u}},{{\hat \omega }_{a,e}},{{\hat \phi }_u}, {{\hat \iota
}_u},{{\hat \beta }_u}} \}_{u=1}^U$ are obtained.

\begin{figure}[t]
\centering
\includegraphics[height=5cm]{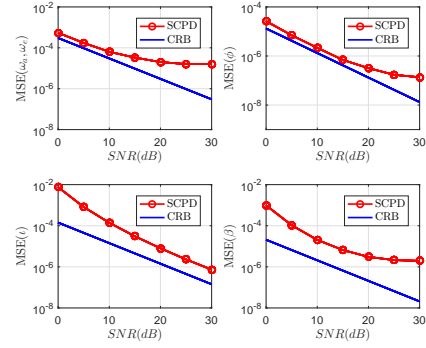}
\caption{MSEs and CRBs of channel parameters vs. SNR.}
\label{fig3}
\end{figure}

\section{Simulation Results}
In this section, we present simulation results to evaluate the
performance of the proposed structured CPD-based (SCPD) method. We
assume that the BS employs a ULA with $N=64$ antennas and $R=1$ RF
chain, the IRS is equipped with $M=16 \times 16$ passive
reflecting elements. In our simulations, the angular parameters
$\left\{\vartheta_{a, l}, \vartheta_{e, l}\right\}_{l=1}^{L}$,
$\left\{\phi_{l}\right\}_{l=1}^{L}$, and $\left\{\chi_{a, l},
\chi_{e, l}\right\}_{l=1}^{L_r}$ are randomly generated from
$[0,2\pi]$, where we set $L=2$ and $L_r=2$. The delay spreads
$\left\{\tau_l\right\}_{l=1}^{L},\left\{\kappa_l\right\}_{l=1}^{L_r}$
are drawn from a uniform distribution ${\cal{U}}(0,100{\rm{ns}})$.
The complex gains
$\left\{\alpha_l\right\}_{l=1}^{L}(\left\{\varrho_l\right\}_{l=1}^{L_r})$
follow a circularly symmetric Gaussian distribution
${\cal{CN}}(0,(c/{4{\pi}D_1f_c})^2)$
$({\cal{CN}}(0,(c/{4{\pi}D_{2,l}f_c})^2))$, where $c$ is the speed
of light, $D_1$ is the distance from the BS to the IRS, $D_{2,l}$
denotes the length of the $l$th path from the IRS to the user, and
$f_c$ is the carrier frequency. We set $D_1=30m$ and
$f_c=28\rm{GHz}$ in our experiments. The total number of
subcarriers is set to $P_0=128$, among which $P$ subcarriers are
used for training. The sampling rate is set to $f_s=0.32\rm GHz$.
The signal-to-noise ratio (SNR) is defined as
\begin{align}
\text{SNR} \triangleq \frac{\|\boldsymbol{\cal { Y
}}-\boldsymbol{\cal{N}}\|_{F}^{2}}{\|\boldsymbol{\cal{N
}}\|_{F}^{2}}
\end{align}

\begin{figure}[!t]
 \centering
\begin{tabular}{cc}
\hspace*{-3ex}
\includegraphics[width=4.5cm,height=4.5cm]{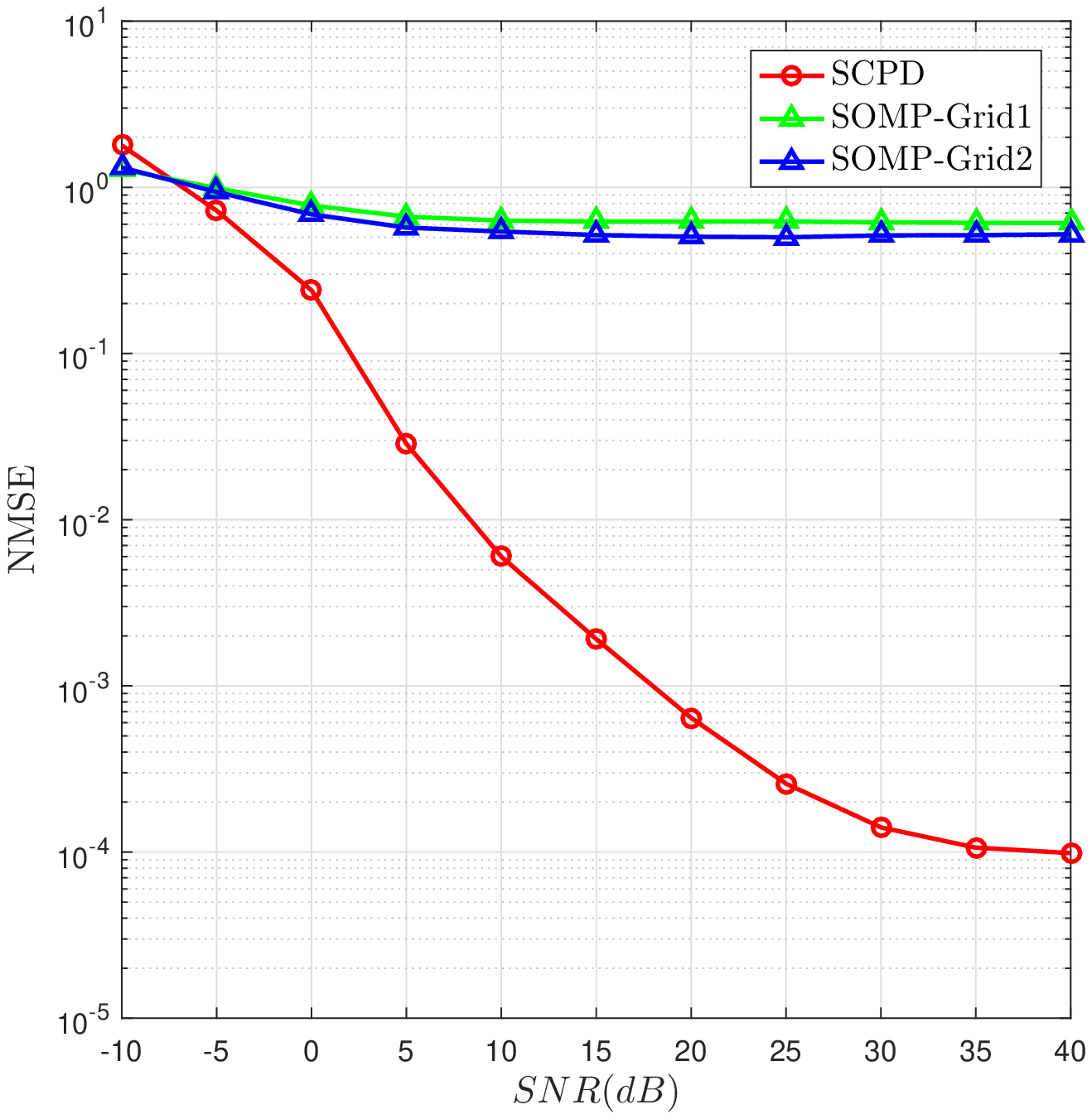}&
\hspace*{-5ex}
\includegraphics[width=4.5cm,height=4.5cm]{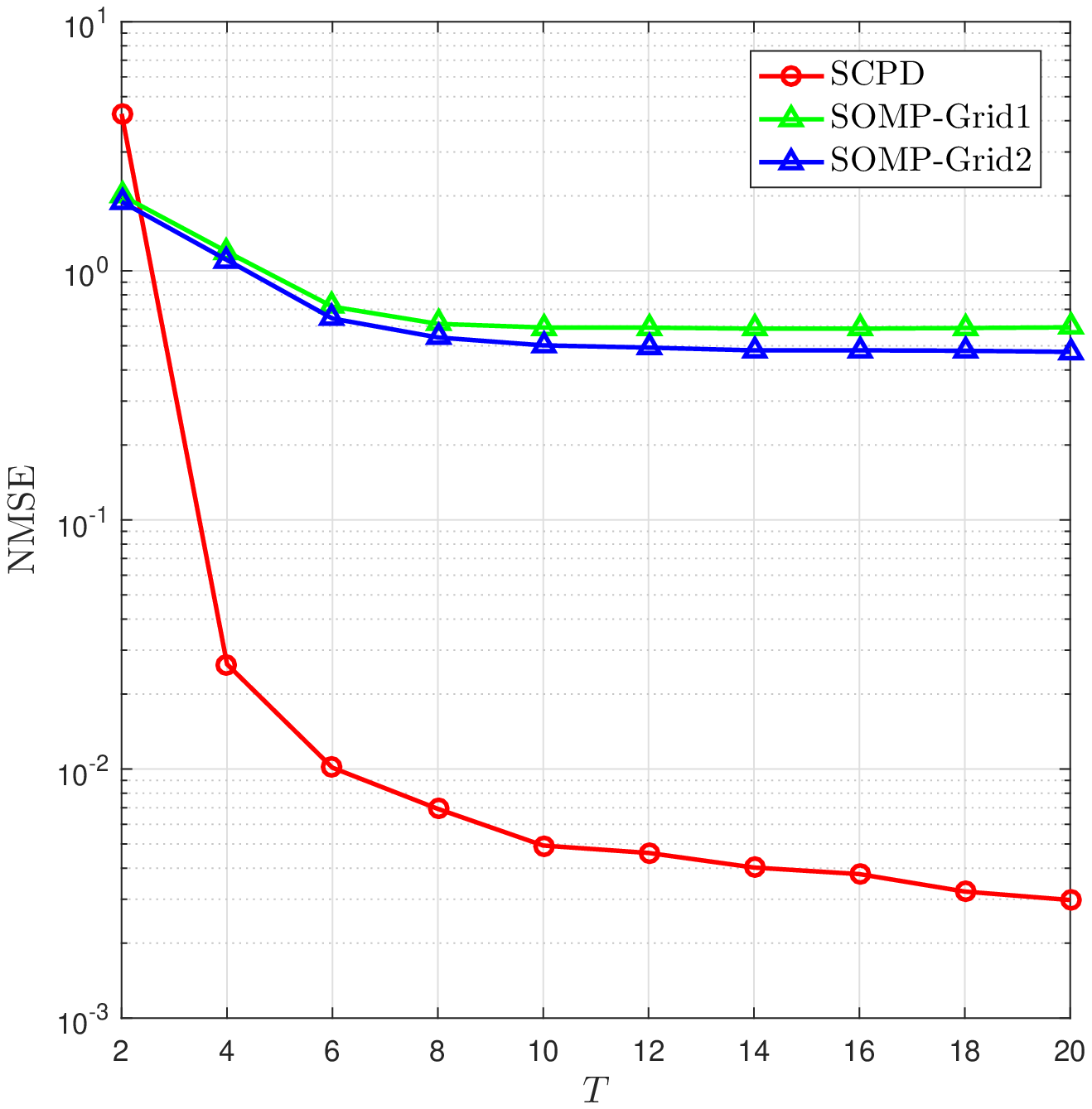}
\\
(a)& (b)
\end{tabular}
  \caption{(a). NMSEs of respective methods vs. SNR, where $P=8$, $T=8$, $Q=8$
  ; (b). NMSEs of respective methods vs. $T$, where $P=8$, $Q=8$, and $\text{SNR}=10$dB.}
   \label{fig4}
\end{figure}

We firstly examine the estimation accuracy of the channel
parameters $\left\{ {{{\omega }_{a,u}},{{ \omega }_{a,e}},{{\phi
}_u},{{\iota }_u},{{\beta }_u}} \right\}_{u = 1}^{ U}$. The CRB
results are also included to provide a benchmark for evaluating
the performance of our proposed method. Note that our estimation
problem has a form similar to that of \cite{ZhouFang17}. Therefore
its CRB can be derived by following the derivations developed in
\cite{ZhouFang17}. In Fig. \ref{fig3}, we depict the mean square
errors (MSEs) of our proposed method versus the SNR, where we set
$P=16$, $T=16$, $Q=16$. From Fig. \ref{fig3}, we see that our
proposed method can obtain accurate estimates of the angular
parameters and the time delays. Its estimation errors are close to
the theoretical lower bound. The estimates of the composite path
gains are not as close to the CRB as other parameters, probably
because the composite path gains are not directly estimated from
the factor matrices.

Next, we report the overall channel estimation performance. By
transforming the channel estimation problem into a MMV compressed
sensing problem, the simultaneous-OMP method (SOMP)
\cite{TroppGilbert06} can also be used to estimate the channel.
For the SOMP, two different grids are employed to discretize the
continuous parameter space: the first grid discretizes the
AoA-AoD-time delay space into $ \left(32 \times 32 \right) \times
128 \times 64$ points, and the second grid discretizes the
AoA-AoD-time delay space into $ \left(64 \times 64 \right) \times
256 \times 128$ points. In Fig. \ref{fig4}, we plot the estimation
performance of respective methods as a function of the SNR and the
number of time frames $T$. The performance is evaluated via the
normalized mean squared error (NMSE) of the cascaded channel,
which is defined as $\sum_{p=1}^{P}\|\boldsymbol{\hat
H}_{p}-\boldsymbol{H}_{p}\|_F^2/\sum_{p=1}^{P}\|\boldsymbol{H}_{p}\|_F^2$.
From these results, we see that the proposed method presents a
substantial performance improvement over the SOMP method. In
addition, we observe that our proposed method can provide reliable
channel estimation when $P=Q=8$ and $T=4$, which corresponds to a
total number of $256$ measurements for training. This result
indicates that the proposed method can achieve a substantial
training overhead reduction.

\section{Conclusion}
In this paper, we developed a CPD-based channel estimation method
for IRS-assisted mmWave OFDM systems. The proposed method exploits
the inherent low-rank structure of the cascade channels and the
inherent Vandermonde structure of the factor matrix. Our analysis
shows that the proposed method only requires a modest amount of
training overhead to extract the channel parameters. Simulation
results were provided to illustrate the efficiency of the proposed
method.


\end{document}